\documentclass[conference]{IEEEtran}
\IEEEoverridecommandlockouts
\usepackage{amsmath,amsfonts}
\usepackage{amssymb,amsmath,amsthm}

\usepackage[letterpaper, top=1.78cm, bottom=2.58cm, left=1.7cm, right=1.7cm]{geometry}

\newtheoremstyle{custom}
  {0.5em} 
  {0.5em} 
  {\itshape} 
  {0.5em} 
  {\itshape\space} 
  {:} 
  {0.5em} 
  {\thmname{#1}\thmnumber{\textit{\hspace{0.5em}#2}}\thmnote{#3}} 

\usepackage{tikz}

\usepackage[flushleft]{threeparttable}

\newtheorem{Definition}{Definition}

\newtheorem{Remark}{Remark}

  {\proof}{\proofend}
\newtheorem{proposition}{Proposition}
\newtheorem{property}{Property}

\makeatletter
\renewenvironment{proof}[1][\proofname]{\par
  \pushQED{\qed}%
  \normalfont \topsep6\p@\@plus6\p@\relax
  \trivlist
  \itemindent1.5em 
  \item[\hskip\labelsep
        \itshape
        #1\@addpunct{:}]\ignorespaces
}{%
  \popQED\endtrivlist\@endpefalse
}
\makeatother

\usepackage{bbm}
\usepackage{algorithm,algorithmic}
\usepackage{subcaption}
\usepackage{array}
\usepackage{textcomp}
\usepackage{stfloats}
\usepackage{url}
\usepackage{verbatim}
\usepackage{graphicx}
\usepackage{cite}
\usepackage{color}
\def\BibTeX{{\rm B\kern-.05em{\sc i\kern-.025em b}\kern-.08em
    T\kern-.1667em\lower.7ex\hbox{E}\kern-.125emX}}

\newcommand{\LSDC}{\mathtt{LS}}
\newcommand{\MSDC}{\mathtt{MS}}
\newcommand{\Nf}{N_\mathrm{f}}

\newcommand{\E}{\mathrm{E}}
\newcommand{\T}{\mathrm{T}}

\newcommand{\p}{\mathrm{p}}
\newcommand{\dd}{\mathrm{d}}

\newcommand{\f}{\mathrm{f}}

\newcommand{\replica}{\mathrm{r}}
\newcommand{\decoding}{\mathrm{c}}

\begin{document}



\title{Network Availability Enhancement in Low-Altitude HetNets: A Cross-Layer Design Perspective}

\author{

\thanks{This work is supported in part by the National Natural Science Foundation of China (Grant No. 62121001, 62495020, and 62461160329), in part by Key Research and Development Program of Shaanxi (Grant No. 2024CY2-GJHX-82),  in part by Fundamental Research Funds for the Central Universities (No. QTZX26093), and in part by 2025 Open Fund Project of the State Key Laboratory of Power Grid Safety (Project Title: Research on Multi-dimensional Resilience Assessment and Enhancement Technologies for Power Systems Considering Primary-Secondary System Integration under Rain, Snow, and Freezing Disasters, No. XTB51202501740). The work of H. Q. Ngo was supported by the UK ISPF programme through the UKRI EPSRC (grant ID: UKRI554, led by University of East Anglia) under pilot project ``BEAM-RAN'' (UEA ref: R213867). The work of M. Matthaiou was supported by the European Research Council (ERC) under the European Unions Horizon 2020 research and innovation  programme (grant agreement No. 101001331).}

\IEEEauthorblockN{
\fontsize{0.36cm}{1cm}\selectfont
Teng~Wu$^\dag$$^\ddagger$, Jiandong~Li$^\dag$, Junyu~Liu$^\dag$, Min~Sheng$^\dag$, Mohammadali Mohammadi$^\ddagger$,  Hien Quoc Ngo$^\ddagger$, and Michail Matthaiou$^\ddagger$}
\IEEEauthorblockA{
$^\dag$State Key Laboratory of ISN, Institute of Information Science, Xidian University, Xi’an, Shaanxi, 710071, China\\
$^\ddagger$Centre for Wireless Innovation (CWI), Queen's University Belfast, Belfast, BT3 9DT, U.K.\\
Email: t.wu@stu.xidian.edu.cn}}


\maketitle

\begin{abstract}
This paper proposes a computing–communication resource interchange method to enhance network availability (NA) in low-altitude heterogeneous networks (LA-HetNets). In these networks, communication resource conflicts and imbalances, caused by extreme heterogeneity (diverse mobility, mixed delays, and hybrid transmission), and cross-regional traffic, reduce reliability and lead to unavailability. Restoring NA requires additional communication resources, yet dynamic cross-regional scheduling is limited, making locally redundant computing resources an alternative to reduce communication resource overhead. While computing resources address medium access control (MAC)-layer unreliability, physical (PHY)-layer functionalities still rely on communication resources. Thus, it remains unclear whether increasing computing resources alone can achieve target NA, especially under greater heterogeneity. We elaborate on the impact of heterogeneity on NA and show that expanding computing resources alone cannot meet target NA under high heterogeneity, as NA degrades sharply due to increased communication capability demands. To overcome this, we propose a cross-layer optimization method enabling computing–communication resource interchange to address both MAC- and PHY-layer unreliability. By reducing processing delays with computing resources while ensuring MAC-layer reliability, our method extends PHY-layer transmission delay and expands communication resources. Simulations demonstrate our approach's superiority in achieving target NA under greater heterogeneity, revealing that computing-communication resource interchange fulfills expanding communication capability demands more effectively than conventional resource overhead reduction.
\end{abstract}
\begin{IEEEkeywords}
Computing-communication resource interchange, cross-layer optimization, low-altitude heterogeneous networks (LA-HetNets).
\end{IEEEkeywords}

\vspace{-0.3em}
\section{Introduction}
\vspace{-0.2em}
Low-altitude networks, employing unmanned aerial vehicle (UAV)-mounted flying access points (FAPs), have emerged as a pivotal paradigm for supporting communications in complex environments \cite{Pan2026TWC}. These networks are essential for a wide range of applications—from smart agriculture and routine industrial inspections to high-stakes situations—where ground infrastructure is unavailable or compromised, e.g., during disaster emergency operations \cite{Ding2026}. However, maintaining a sustained NA remains difficult, as multiple fundamental constraints introduce significant unreliability in communication, where the NA is the probability that the quality-of-service (QoS) requirements, in terms of delay and reliability, for each service are satisfied \cite{LiuTCOM2025, Wu2026TCOM}. On one hand, the inherent mobility of both FAPs and user equipment (UEs), combined with wireless fading, creates highly dynamic channels, which deteriorates PHY-layer link reliability (e.g., decoding errors and transmission outage) \cite{Pan2026TWC, Wu2026TCOM}. On the other hand, the network must accommodate heterogeneity induced by diverse UE mobility profiles, mixed delay constraints, and hybrid transmission modes, such as unicast for ultra-reliable and low-latency communication and multicast for other delay-sensitive traffic \cite{Ding2026, Yao2022MNET}. This heterogeneity results in communication resource allocation conflicts \cite{Alsenwi2019CL}. Moreover, traffic in low-altitude networks exhibits cross-regional characteristics, triggering severe communication resource allocation imbalances \cite{Peng2021JSAC}. Collectively, these resource allocation conflicts and imbalances  lead to communication reliability degradation \cite{Alsenwi2019CL, Peng2021JSAC}.

To address unreliability and resource allocation conflicts, a unified broadcast scheme can streamline scheduling, while coordinated multi-point (CoMP) and multi-connectivity (MC) techniques leverage spatial and frequency diversity to enhance the stability of wireless links \cite{Pan2026TWC, Li2024}. To circumvent the prohibitive costs of cross-regional communication resource scheduling, exploiting locally abundant computing resources, i.e., central processing unit (CPU) cycles, emerges as a viable alternative to address communication resource allocation imbalances \cite{Peng2021JSAC}. Conventionally, mobile edge computing leverages such local computing resources for data compression or task offloading to alleviate communication resource pressure \cite{Peng2021JSAC}. However, when communication resources are scarce, such strategies fail to address PHY-layer unreliability, which inherently requires sufficient communication resources. 

Therefore, whether adding computing resources can achieve the target NA in LA-HetNets remains a critical open problem, particularly under greater heterogeneity, directly motivating our work. The main contributions of this paper are as follows:
\begin{itemize}
\vspace{-0.25em}
    \item To address this open problem, we first provide analytical expressions that capture the impact of heterogeneity on NA for FAP CoMP-enabled LA-HetNets under a unified MC broadcast scheme. Our analysis reveals that adding computing resources can resolve MAC-layer unreliability to enhance NA. However, simply scaling up computing resources cannot achieve the target NA under increasing heterogeneity. This is because, under the unified scheme, greater heterogeneity inflates the per-link service load, thereby increasing communication capability demands.
    \item To break the above bottleneck, we propose a cross-layer optimization to enable computing-communication resource interchange. Specifically, reducing processing delay by allocating additional computing resources—while maintaining MAC-layer reliability— extends the PHY-layer transmission delay, effectively increasing communication resources. Simulations demonstrate that our approach successfully achieves the target NA, supporting greater heterogeneity than conventional computing-only expansion strategies. This demonstrates that trading computing resources for communication resources is more effective than merely alleviating communication resource constraints through conventional strategies in meeting communication capability demands. 
\end{itemize}  
\vspace{-0.1em}
\textit{Notation:} Bold lowercase letters denote vectors; $(\cdot)^T$ and $(\cdot)^H$ represent the transpose and Hermitian transpose, respectively; ${f_{{Q^{ - 1}}}}(  \cdot  )$ is the inverse Q-function, while the Q-function is ${f_Q}( x ) = \frac{1}{{\sqrt {2\pi } }}\int_x^\infty  {\exp ( { - \frac{{{y^2}}}{2}} )dy} $; $\mathbbm{1}[ Y ]$ denotes the indicator function of event $Y$, where $\mathbbm{1}[ Y ]$ = 1 if event $Y$ is true, and $\mathbbm{1}[ Y ]$ = 0 otherwise; $\mathcal{CN}(0,\sigma^2)$ denotes a circularly symmetric complex Gaussian random variable (RV) with variance $\sigma^2$. Finally, $\mathbb{E}_X\{\cdot\}$ and $\mathbb{P}_X\{\cdot\}$ denote the statistical expectation and probability with respect to the RV $X$, respectively.

\section{System Model}\label{Sec_System_Model}
\vspace{-0.2em}
\subsection{Network Model}
\vspace{-0.2em}
We consider a downlink FAP CoMP-enabled LA-HetNet, which consists of $L$ fixed-wing UAV-mounted single-antenna FAPs and $M$ single-antenna UEs with diverse mobility profiles. All FAPs are connected to a data center-level computing unit (DCCU) via wireless fronthaul links for centralized signal processing, while the DCCU can be deployed on the ground or in the sky \cite{Wu2026TCOM}. The service area ${\mathcal A}$ is modeled as a circular region centered at $(0,0,0)$ with a radius $W_{\mathrm{D}}$. The UEs include ground UEs and aerial UEs, where the former are randomly distributed within ${\mathcal A}$, while the latter are randomly distributed over the altitude range ${\mathcal H} \!=\! \left[ {{h_{\min }},{h_{\max }}} \right]$ above ${\mathcal A}$. Moreover, ${h_{\min }}$ and ${h_{\max }}$ are the minimum and maximum altitudes, wherein aerial UEs are located, respectively. All FAPs follow periodic circular flight trajectories \cite{Wu2026TCOM}. The flight altitude, radius, and speed of each FAP are $h_{\mathrm{F}}$, $W_{\mathrm{F}}$, and $v_{\mathrm{F}}$, respectively. The UE mobility in LA-HetNets is characterized by the maximum relative speed (${v_{\mathrm{r}}}$) between the UE and FAPs, which varies across different UEs. The flight period of each FAP is $T \!=\! {{2\pi {W_{\mathrm{F}}}} \mathord{\left/ {\vphantom {{2\pi {W_{\mathrm{F}}}} {{v_{\mathrm{F}}}}}} \right. \kern-\nulldelimiterspace} {{v_{\mathrm{F}}}}}$. To facilitate a tractable performance analysis, we equally divide $T$ into $N_{\mathrm{T}}$ time slots with a duration ${T_{\mathrm{S}}}$, yielding $N_{\mathrm{T}} \!=\! {T \mathord{\left/ {\vphantom {T {{T_{\mathrm{S}}}}}} \right. \kern-\nulldelimiterspace} {{T_{\mathrm{S}}}}}$. This discrete-time division ensures that the transmission distance $d = \sqrt {d_{\mathrm{H}}^2 + {d_{\mathrm{V}}^2}} \le {\tilde d}$ from UEs to FAPs can be considered constant within each time slot and varies only between different slots \cite{Peng2021JSAC, Wu2026TCOM}, where ${d_{\mathrm{H}}}$ and $d_{\mathrm{V}}$ denote the horizontal and vertical distances, respectively. Moreover, ${\tilde d}  \!\! = \!\! \sqrt  {\!{{\tilde d_{\mathrm{H}}^2}} \!+\! {{{\tilde{d}_{\mathrm{V}}^2}}}}$ is the farthest transmission distance from UEs to FAPs, where $\tilde d_{\mathrm{H}}  \! = \! 2{W_{\mathrm{D}}} + 2{W_{\mathrm{F}}}$ and ${\tilde{d}_{\mathrm{V}}}  \! = \! \max \left\{ {h_{\mathrm{F}},\left| {h_{\mathrm{F}} \!-\! {h_{\min }}} \right|,\left| {h_{\mathrm{F}} \!-\! {h_{\max }}} \right|} \right\}$ are the farthest horizontal and vertical distances, respectively.

\subsection{Transmission Scheme and Traffic Scheduling}\label{sub_sec_transmission}
Low-altitude network traffic experiences hybrid transmission modes and mixed delay constraints \cite{Yao2022MNET, Ding2026}, comprising unicast services under more stringent delay constraints (MSDC) and multicast services under less stringent delay constraints (LSDC). Inspired by \cite{Li2024}, we adopt a unified MC broadcast scheme to facilitate the satisfaction of mixed delay constraints. As shown in Fig. \ref{transmission_scheme}, the DCCU concatenates all packets and forwards them to all FAPs, and each FAP then broadcasts these aggregated packets to all UEs over the same time-frequency resources. For downlink broadcasting, the DCCU performs closed-loop precoding utilizing channel state information (CSI) acquired from uplink training \cite{Sadeghi2018TWC}. The MC transmission scheme with packet duplication is used to improve reliability, where packet replicas are transmitted to UEs over independent subchannels. Note that each subchannel bandwidth ($B$) is partitioned as $B = B_{\LSDC} + B_{\MSDC}$, where $B_{\varsigma}$, $\varsigma \in \{\LSDC, \MSDC\}$ denotes the bandwidth allocated to the services under LSDC and MSDC, respectively.


\begin{figure}[t]
\centering
\includegraphics[width=3.2in]{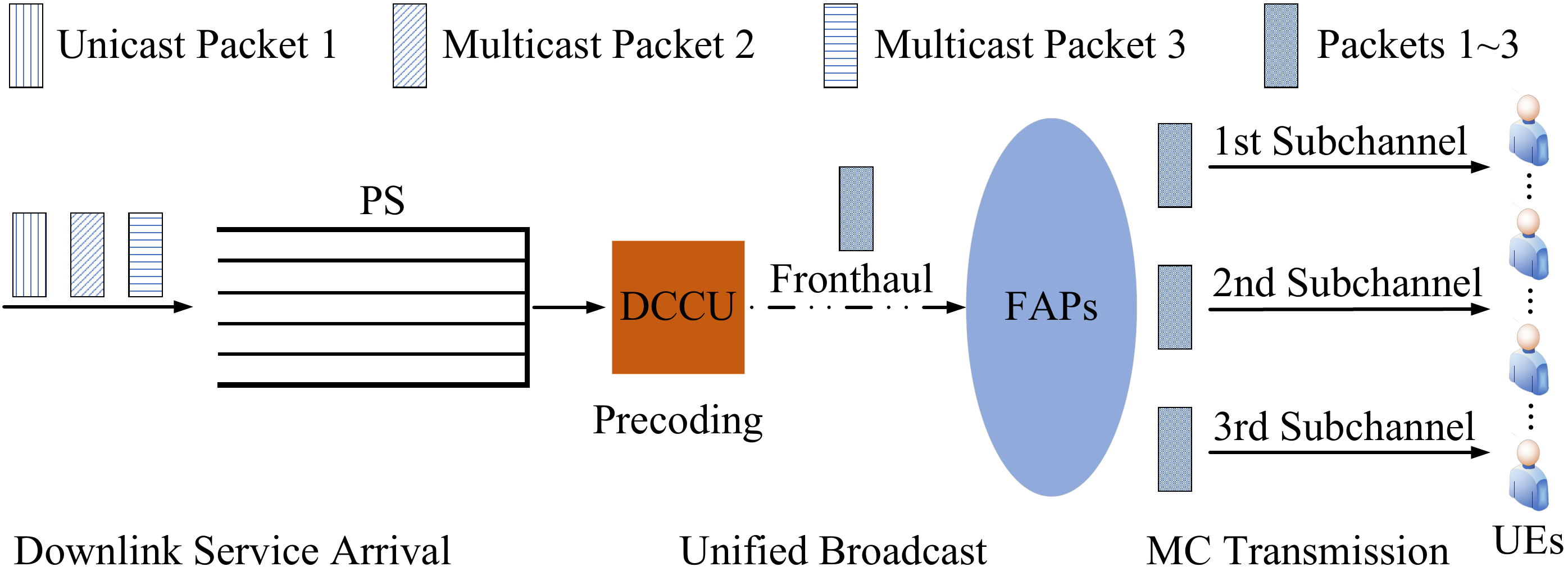}
\caption{The unified MC broadcast scheme and PS server model.} 
\label{transmission_scheme}
\end{figure}

To handle the significantly expanded transmission volume from the unified MC broadcast scheme and guarantee mixed delay constraints, we adopt a processor sharing (PS) server model in the DCCU. By time-sharing computing resources, the PS server immediately processes massive arriving packets, avoiding prolonged head-of-line blocking \cite{She2019IoTJ}. Due to the complex, unknown aggregate downlink arrivals of varying traffic, we employ a generalized G/G/1-PS server model \cite{She2019IoTJ}, where the first and second ``G" denote arbitrary distributions for the service arrival process and required CPU cycles per packet, respectively, while ``1-PS" indicates a single PS server in the buffer. For each UE under the unified scheme, the service load and burstiness are characterized by the average arrival rate $\theta_{\varsigma}$ (in packets/frame) and variance $\sigma_{\varsigma}^2$ (in packets$^2$/frame$^2$) \cite{Wu2026TCOM}. Note that $\theta_{\varsigma}$ and $\sigma_{\varsigma}^2$ are the respective sums of the average arrival rates and variances for services under mixed delay constraints and hybrid transmission modes. For the PS server, the required CPU cycles per bit, $\Omega_{\p}$, is a RV following an arbitrary distribution with mean $\bar{\Omega}_{\p}$ and variance $\sigma_{\p}^2$. Finally, $\Omega_{\mathrm{R}}$ (in CPU cycles/frame) is the DCCU processing rate that denotes the computing resources, while $\Omega_{\mathrm{O}}$ (in CPU cycles/packet) is the additional processing overhead in practical PS server models \cite{She2019IoTJ, Li2024}.\footnote{Analytical formulations use packets/frame and CPU cycles/frame, whereas simulations report their per-second equivalents for practical feasibility.}

\subsection{Channel and Signal Models}
In the proposed design, we divide time-frequency resources into coherence intervals based on the bandwidth $B$ and time slot $T_{\mathrm{S}}$ to approximate the channel as quasi-static and frequency-flat \cite{Sadeghi2018TWC}. Specifically, $B$ is restricted by the coherence bandwidth $B_{\mathrm{C}}$, while $T_{\mathrm{S}}$ equals the transmission time interval (TTI), which is the minimal network time granularity.\footnote{In low-altitude scenarios with a carrier frequency of 2 GHz and a maximum relative velocity of 50 m/s, the coherence time ($T_{\mathrm{C}} \approx 3$ ms) is much greater than TTI (${T_{\mathrm{S}}} = 0.1$ ms) \cite{Tse2005}.} \par


Under the realistic constraints of imperfect CSI, we consider the widely-adopted co-pilot strategy to estimate the composite channel, which is a linear combination of the individual channels of the broadcast UEs \cite{Sadeghi2018TWC}. Based on minimum-mean-square error (MMSE) channel estimation and maximum ratio transmission (MRT) precoding, and under the practical assumption that the receiver relies solely on the statistical channel mean for signal decoding, the received signal at the $m$th UE over any coherence interval can be written as \cite{Sadeghi2018TWC}
\begin{align}
    y_m = \sqrt {{p_{\mathrm{t}}}} [ {{\mathbb{E}_\psi }\{ {{\mathbf{g}}_m^H{\mathbf{v}}} \} + {\mathbf{g}}_m^H{\mathbf{v}} - {\mathbb{E}_\psi }\{ {{\mathbf{g}}_m^H{\mathbf{v}}} \}} ]{s_\varsigma } + {z_\varsigma },
\end{align}
where $p_{\mathrm{t}}$ is the transmit power of each FAP; $s_{\varsigma}$ is the data symbol broadcasted to each UE with $\mathbb{E}_{s}\{|s_{\varsigma}|^2\} = 1$; ${{\bf{g}}_m} = [ {{g_{1,m}}, \ldots ,{g_{L,m}}} ]^{T} \in {\mathbb{C}^{L \times 1}}$ is the CSI vector from $L$ FAPs to the $m$th UE; ${g_{l,m}} = \sqrt {\beta( d_{l,m} )} {\psi_{l,m}}$ is the channel from the $l$th FAP to the $m$th UE; $\beta( d_{l,m} )$, ${\psi_{l,m}}$, $d_{l,m}$ represent the path loss, the effect of small-scale fading and shadowing, and transmission distance between the $l$th FAP and the $m$th UE, respectively; $\beta( d_{l,m} ) = \mathcal{A}d_{l,m}^{ - 2 }$, where $\mathcal{A}( {{\mathrm{dB}}} ) =  - 20{\log _{10}}(4\pi / \lambda )$ denotes the path loss at 1 m with wavelength $\lambda = {{{v_{\mathrm{L}}}} \mathord{\left/
{\vphantom {{{v_{\mathrm{L}}}} {{f_{\mathrm{c}}}}}} \right.
\kern-\nulldelimiterspace} {{f_{\mathrm{c}}}}}$, with ${f_{\mathrm{c}}}$ being the carrier frequency, and ${v_{\mathrm{L}}}$ being the speed of light; ${\psi _{l,m}}$ is modeled via the $\kappa$-$\mu$ shadowed fading model with parameters $\kappa$, $\mu$, and $\bar{m}$ \cite{LiuTCOM2025}; $z_{\varsigma} \sim {\mathcal C}{\mathcal N}\left( {0,B_{\varsigma}N_0} \right)$ is the additive white Gaussian noise with $N_0$ being the noise power spectral density; ${\mathbf{v}} = \frac{{{{\bf{\hat g}}}}}{{\sqrt {\mathbb{E}_\psi\{ {{{\| {{\mathbf{\hat g}}} \|}^2}} \}} }} \in {\mathbb{C}^{L \times 1}}$ is the MRT precoding vector; ${\bf{\hat g}} \!=\! [ {{{\hat g}_1}, \ldots ,{{\hat g}_L}} ]^{T} \in {\mathbb{C}^{L \times 1}}$ is the estimated composite channel vector from $L$ FAPs to $M$ UEs; ${g_l}{\mathrm{ = }}\sum\nolimits_{m = 1}^M {{g_{l,m}}} $ and ${{\hat g}_l}{\mathrm{ = }}\frac{{\sum\nolimits_{m = 1}^M {\beta \left( {{d_{l,m}}} \right)} }}{{\sum\nolimits_{m = 1}^M {\beta \left( {{d_{l,m}}} \right)}  + 1/\left( {{\rho _{\mathrm{u}}}\xi } \right)}}\left( {{g_l} + {w_l}} \right)$ are the composite channel and estimated composite channel between the $l$th FAP and $M$ UEs, respectively; ${w_l} \sim \mathcal{CN}(0,{\frac{1}{{{\rho _{\mathrm{u}}}\xi }}})$ denotes the additive receiver noise after being projected onto the pilot sequence \cite{Sadeghi2018TWC}, while the channel estimation error is ${\hat w_l} = {g_l} - {\hat g_l}$. Finally, $\rho_{\mathrm{u}} = p_{\mathrm{u}}/(B_{\varsigma}N_0)$ is the transmit signal-to-noise ratio (SNR) in uplink channel training, where $p_{\mathrm{u}}$ is the transmit power of the pilots and $\xi$ is the pilot length.

Then, under realistic wireless fronthaul, the effective SNR at the $m$th UE over any coherence interval is \cite{Sadeghi2018TWC}
\begin{align}\label{eq_achievable_SNR}
    \!\!{\gamma^{\varsigma} _{m,\E}}(M,B_{\varsigma}) \!=\! \frac{{{\phi _{{\mathrm{wf}}}}p_{\mathrm{t}}{{\big| {\mathbb{E}_\psi\{ {{\mathbf{g}}_m^H{\mathbf{v}}} }\} \big|}^2}}}{{p_{\mathrm{t}}\mathbb{E}_\psi\{ {{{| {{\mathbf{g}}_m^H{\mathbf{v}}} |}^2}} \} \!-\! p_{\mathrm{t}}{\big|{ {\mathbb{E}_\psi\{ {{\mathbf{g}}_m^H{\mathbf{v}}} \}} }\big|^2} \!+\! B_{\varsigma}{N_0}}}\!,\!
\end{align}
where $\phi_{\mathrm{wf}} \in ( {0,1} )$ is the SNR loss due to wireless fronthaul.

\section{Preliminary Analysis}
\subsection{QoS Requirements}\label{sub_QoS_requirement}
\vspace{-0.1cm}
In downlink, QoS requirements of each service in terms of delay and reliability are represented by the total downlink delay bound $D_{\max}^{\varsigma}$ and the target downlink overall packet loss (dOPL) probability $\varepsilon_{\max}^{\varsigma}$, respectively \cite{Li2024}. According to 3GPP, downlink delay includes transmission, queuing, processing, propagation, and routing delays \cite{Li2024}. Since routing and propagation delays are deterministic, and head-of-line queuing is avoided by the PS server model \cite{She2019IoTJ}, $D^{\varsigma}$ is dominated by the transmission delay ($D_{\dd}^{\varsigma}$) and processing delay ($D_{\p}^{\varsigma}$). Then, the delay constraint is satisfied by:
\begin{equation}\label{eq_delay_constraint}
D^{\varsigma} = D_{\dd}^{\varsigma} + D_{\p}^{\varsigma} \le D_{\max }^{\varsigma}.
\end{equation}
A short frame structure is employed for satisfying delay constraints \cite{Dong2021TWC, Li2024}, where the frame duration $T_{\f} \!=\! T_{\f}^{(\dd)} + T_{\f}^{(\mathrm{c})}$ equals the TTI. Here, $T_{\f}^{(\dd)}$ and $T_{\f}^{(\mathrm{c})}$ are the data transmission and control signaling (orthogonal pilots for CSI estimation) durations, respectively. Then, the pilot length can be expressed as $\xi = BT_{\f}^{({\mathrm{c}})}$. Within the delay constraints, we define the total number of frames as ${N_{\f}^{\varsigma}}(D_{\max }^{\varsigma}) \!=\! \frac{{D_{\max }^{\varsigma}}}{{{T_{\f}}}}$. Similarly, the numbers of frames occupied by the transmission delay and processing delay are denoted by $N_{\f}^{\dd}( {{D_{\dd}^{\varsigma}}} ) \!=\! \frac{{{D_{\dd}^{\varsigma}}}}{{{T_{\f}}}} \ge 1$ and $N_{\f}^{\p}( {{D_{\p}^{\varsigma}}} ) \!=\! \frac{{{D_{\p}^{\varsigma}}}}{{{T_{\f}}}} \ge 1$, respectively.  \par

Following \cite{Wu2026TCOM}, packets with decoding errors are discarded. Given $D_{\dd}^{\varsigma} + D_{\p}^{\varsigma} = D_{\max}^{\varsigma}$, delayed packets are also discarded due to transmission outages (failing to transmit within $D_{\dd}^{\varsigma}$ \cite{Tse2005}) or processing delay violations (failing to process within $D_{\p}^{\varsigma}$ \cite{She2019IoTJ}). Thus, to effectively reflect reliability, the dOPL probability $\varepsilon^{\varsigma}$ accounts for the transmission failure probability $\varepsilon_{\dd}^{\varsigma}$ (due to decoding errors and outages) and the processing delay violation probability $\varepsilon_{\p}^{\varsigma}$. Then, the reliability requirement is ensured by:
\begin{equation}\label{eq_reliability_constraint}
\varepsilon^{\varsigma} = 1 - (1 - \varepsilon_{\p}^{\varsigma})(1 - \varepsilon_{\dd}^{\varsigma}) \le \varepsilon_{\p}^{\varsigma} + \varepsilon_{\dd}^{\varsigma} \le \varepsilon_{\max}^{\varsigma}.
\end{equation}

\subsection{Wireless Link Quality}
For delay-sensitive services, the blocklength of channel coding is finite in practice \cite{Dong2021TWC}. Then, with the effective SNR $\gamma$, the achievable service rate (packet/slot) under LSDC is \cite{Dong2021TWC} 
\begin{equation}\label{eq_Rate_LDC}
{R^{{\LSDC}}} (B_{\LSDC}, \gamma) = \frac{B_{\LSDC}T_{\f}^{(\dd)}}{{{\varpi ^{{\LSDC}}}\ln 2}}C\left( {\frac{\gamma }{\phi _{{\mathrm{fb}}}}} \right),
\end{equation}%
where ${\varpi ^{{\LSDC}}}$ is the packet size (in bit) of the services under LSDC; $C\left( \gamma  \right) = \ln \left( {1 + \gamma } \right)$ is Shannon’s capacity in the infinite blocklength regime, while ${\phi _{{\mathrm{fb}}}} > 1$ is the SNR gap due to finite blocklength (FBL).  
\par 

When the delay constraints are more stringent, the blocklength of channel coding is significantly shorter \cite{Dong2021TWC}. As a result, decoding errors cannot be ignored under MSDC. According to FBL information theory, the achievable service rate $R^{\MSDC}$ (packet/slot), with the effective SNR $\gamma $ and a given decoding error probability ${\varepsilon_{\decoding}^{\MSDC}}$, is expressed as \cite{Dong2021TWC, Li2024} 
\begin{equation}\label{eq_rate_MDC}
\!\!{R^{\MSDC}}(B_{\MSDC}, \gamma) = \frac{B_{\MSDC}T_{\f}^{(\dd)}}{{{\varpi ^{\MSDC}}\ln 2}}\bigg[ {C( \gamma  )- \sqrt {\frac{{V( \gamma )}} {{B_{\MSDC}T_{\f}^{(\dd)}}}} {f_{{Q^{ - 1}}}}( {{{\varepsilon_{\decoding}^{\MSDC}}}} )} \bigg],\!
\end{equation}
where ${\varpi ^{{\MSDC}}}$ is the packet size (in bit) of the services under MSDC and $V( \gamma  ) = 1 - \frac{1}{{{{( {1 + \gamma } )}^2}}}$ is the channel dispersion.

Based on \eqref{eq_Rate_LDC} and \eqref{eq_rate_MDC}, ${\varepsilon_\dd^{\LSDC}}$ comes from transmission outage, and ${\varepsilon_\dd^{\MSDC}}$ comes from decoding error and transmission outage. Thus, the decoding error probability under LSDC is ${\varepsilon_{\decoding}^{\LSDC}} = 0$ and the decoding error probability under MSDC is ${\varepsilon_{\decoding}^{\MSDC}} > 0$.

\subsection{NA Definition}
The NA is defined as the probability that the QoS requirements, in terms of delay and reliability, for each UE's service are satisfied \cite{LiuTCOM2025, Wu2026TCOM}. In the service coexistence scenario under LSDC and MSDC, we partition the total UE set ${\mathcal M} = \{ {1, \ldots ,M} \}$ into two subsets: ${\mathcal M}_{\LSDC}$ subject to the LSDC and ${\mathcal M}_{\MSDC}$ subject to the MSDC, i.e., ${\mathcal M}_{\LSDC} \subseteq {\mathcal M}$ and ${\mathcal M}_{\MSDC} \subseteq {\mathcal M}$. Specifically, ${\mathcal M}_{\LSDC}$ consists of $M_{\LSDC}$ UEs and ${\mathcal M}_{\MSDC}$ consists of $M_{\MSDC}$ UEs, satisfying $M_{\LSDC} + M_{\MSDC} = M$. Thus, the NA can be expressed as 
\begin{equation}\label{eq_NA_define}
\begin{aligned}
 \eta   \buildrel \Delta \over  
 =   & \frac{1}{M}  \bigg( \sum\nolimits_{m \in {\mathcal M}_{\MSDC}}  {\mathbbm{1}\left[ {{{\left( {{D^{\MSDC }} \le D_{\max }^{\MSDC},{ \varepsilon ^{\MSDC}} \le \varepsilon _{\max }^{  \MSDC }} \right)}_m}} \right]} \\
& + \sum\nolimits_{m' \in {\mathcal M}_{\LSDC}} {\mathbbm{1}\left[ {{{\left( {{D^{\LSDC }} \le D_{\max }^{\LSDC},{ \varepsilon ^{\LSDC}} \le \varepsilon _{\max }^{  \LSDC }} \right)}_{m'}}} \right]} \bigg), 
\end{aligned}
\end{equation}
where ${{{\left( {{D^{\MSDC }} \!\le\! D_{\max }^{\MSDC },{ \varepsilon ^{\MSDC }} \!\le\! \varepsilon _{\max }^{  \MSDC }} \right)}_m}}$ is an event that the delay and reliability requirements for the communication service under MSDC of $m$th UE are satisfied ($\forall m \in {\mathcal M}_{\MSDC}$), while ${{{\left( {{D^{\LSDC }} \!\le\! D_{\max }^{\LSDC },{ \varepsilon ^{\LSDC }} \!\le\! \varepsilon _{\max }^{  \LSDC }} \right)}_{m'}}}$ is an event that the delay and reliability requirements for the communication service under LSDC of $m'$th UE are satisfied ($\forall m' \in {\mathcal M}_{\LSDC}$). \par

Given a target NA ${\eta _{\max }}$, achieving $\left( {{\eta } \ge {\eta _{\max }},{\eta _{\max }}  \to 1} \right)$ indicates that each service of UEs is ensured.

\enlargethispage{-0.03in}

\section{NA Analysis and Enhancement}

\subsection{NA Analysis}
Although NA evaluation in LA-HetNets inherently requires accounting for spatio-temporal dimensions due to FAP and UE mobility \cite{Wu2026TCOM}, we focus on the NA boundary for analytical tractability and robust network design. Similar to \cite{Wu2026TCOM}, we consider a worst-case scenario where all UEs are activated at the maximum transmission distance ($d = \tilde{d}$). The complex spatio-temporal dynamics are abstracted away under the worst-case scenario. Consequently, the heterogeneity affects the NA boundary solely through the mixed delay constraints. Building on this, we subsequently define the degree of heterogeneity and derive an explicit expression for the equivalent NA to capture its influence and NA boundary. For the convenience of analysis, the $m$th UE can be treated as a typical UE. The effective SNR of the typical UE is denoted as ${\tilde \gamma^{\varsigma} _{\T,\E}}(M,B_{\varsigma})$, which is obtained by substituting $d = \tilde{d}$ into \eqref{eq_achievable_SNR}.

\begin{Definition}\label{Def_heterogeneity}
Assume that $U$ is the degree of heterogeneity. 
There are $U$ distinct delay constraints across the two UE sets, comprising $U_{\LSDC}$ constraints for set ${\mathcal M}_{\LSDC}$ and $U_{\MSDC}$ constraints for set ${\mathcal M}_{\MSDC}$ with $U = {U_{\LSDC}} + {U_{\MSDC}}$. We aggregate these $U$ unique delay constraints to formulate a system heterogeneity matrix, $\mathbf{\Theta} \in \mathbb{R}^{2 \times U}$:
\begin{equation}
    \mathbf{\Theta} \triangleq \begin{bmatrix}
    ( D_{\dd}^\LSDC )_1, \cdots , ( D_{\dd}^\LSDC )_{U_{\LSDC}}, ( D_{\dd}^\MSDC )_1,   \cdots  , ( D_{\dd}^\MSDC )_{U_{\MSDC}} \\
    ( D_{\p}^\LSDC )_1,  \cdots  , ( D_{\p}^\LSDC )_{U_{\LSDC}}, ( D_{\p}^\MSDC )_1,  \cdots , ( D_{\p}^\MSDC )_{U_{\MSDC}}
    \end{bmatrix},
\end{equation}
where
\begin{equation}\label{eq_delay_constraints_components}
   \! (D_{\max}^\varsigma )_{u_{\varsigma}} \!=\! (D_{\dd}^\varsigma )_{u_{\varsigma}} \!+ (D_{\p}^\varsigma )_{u_{\varsigma}}, \forall u_{\varsigma} \!\in [1, U_{\varsigma}], \varsigma \!\in \{\LSDC, \MSDC\}.
\end{equation}
\end{Definition}

\begin{figure*}[t]
\centering
\begin{align}\label{eq_NA_all_REC}
   \eta _{{\mathrm{E}}}^{\mathrm{H}}( {U,\mathbf{\Theta},N_{\replica} },{\Omega _{\mathrm{R}}} )  
    = \prod\limits_{\varsigma  \in \{ {\LSDC,\MSDC} \}} { \min \bigg(1, \frac{R^{\varsigma}(B_{\varsigma},{\tilde \gamma^{\varsigma} _{\T,\E}}(M,B_{\varsigma}))}{ \max \big\{ {R_{\mathrm{th}}^{\varsigma}\big( (D_{\p}^\varsigma)_{u_{\varsigma}},(D_{\dd}^\varsigma)_{u_{\varsigma}}, N_{\replica}, {\Omega _{\mathrm{R}}} \big)}, \forall u_{\varsigma} \in [1, U_{\varsigma}] \big\} } \bigg)   }    .
\end{align}
\vspace{-1em}
\hrulefill
\end{figure*}

\begin{proposition}\label{Proposition_LB_NA}
Under the unified MC broadcasting scheme with $N_{\replica}$ subchannels for transmitting packet replicas, the equivalent NA in FAP CoMP-enabled LA-HetNets is expressed in \eqref{eq_NA_all_REC} at the top of the next page. Note that $R_{\mathrm{th}}^{\varsigma}(D_{\p}^{\varsigma}, D_{\dd}^{\varsigma},N_{\replica}, {\Omega _{\mathrm{R}}})$ denotes the threshold of service rate to satisfy the QoS requirements ($D_{\max}^{\varsigma}, \varepsilon_{\max}^{\varsigma}$) that is given by
\begin{align}\label{eq_rate_threshold_MSDC_LSDC}
& R_{\mathrm{th}}^{\varsigma}(D_{\p}^{\varsigma}, D_{\dd}^{\varsigma}, N_{\replica}, {\Omega _{\mathrm{R}}})
     =  \frac{{N_{\f}^{\varsigma}}(D_{\max }^{\varsigma})}{N_{\f}^{\dd}(D_{\dd}^{\varsigma})} \\
    & \times \! \Bigg[{{\theta _{\varsigma}} \!+\! \Bigg(\!\frac{{\sigma _{\varsigma}^2}}{{N_{\f}^{\varsigma}}(D_{\max }^{\varsigma})\big({{{\big( {{\varepsilon^{\varsigma} _{\max }} - {\hat \varepsilon_{\p}^{\varsigma}}( {{D_{\p}^{\varsigma}}},N_{\replica},{\Omega _{\mathrm{R}}} ) } \big)}^{\frac{1}{{{N_{\replica}}}}} - \varepsilon_{\decoding} ^{\varsigma} }}\big)}}\!\Bigg)^{\!\!\frac{1}{2}\!}\Bigg]\!. \notag
\end{align}
Finally, ${\hat \varepsilon_{\p}^{\varsigma}}( {{D_{\p}^{\varsigma}}}, N_{\replica},{\Omega _{\mathrm{R}}} )$ is the upper bound (UB) on the processing delay violation probability, expressed as
\begin{align}\label{eq_processing_delay_prob_UB}
    & {\hat \varepsilon_{\p}^{\varsigma}}( {{D_{\p}^{\varsigma}}}, N_{\replica},{\Omega _{\mathrm{R}}} ) =
     \frac{1}{{{{\big( {\frac{{\Omega _{\mathrm{R}}}}{N_{\replica}} - ({\theta _{\LSDC}} + {\theta _{\MSDC}})( {{{\bar \Omega }_{\p}}{\varpi ^{{\varsigma}}} + {\Omega _{\mathrm{O}}}} )} \big)}^2}}} \\
    & \times \!\! \bigg[\! \sigma _{\p}^2{\big(\theta _{\LSDC}^2 + \theta _{\MSDC}^2\big)} \!+\! \frac{{(\sigma _{\LSDC}^2 + \sigma _{\MSDC}^2) \big( \sigma _{\p}^2{\varpi ^{{\varsigma}}} + {{( {{{\bar \Omega }_{\p}}{\varpi ^{{\varsigma}}} + {\Omega _{\mathrm{O}}}} )}^2}\big)}}{N^{\p}_{\f}(D_{\p}^{\varsigma})}  \!\bigg]\!.\notag
\end{align}
\end{proposition}

\begin{proof}
    See Appendix \ref{appendix_Proposition_LB_NA}.
\end{proof}

\begin{Remark}
    The target NA achievement, i.e., $( {{\eta } \ge {\eta _{\max }},{\eta _{\max }}  \to 1} )$, is strictly guaranteed if the equivalent NA in \eqref{eq_NA_all_REC} attains its theoretical limit value of $1$, i.e., $\eta_{\mathrm{E}}^{\mathrm{H}} = 1$.
\end{Remark}

\noindent
From \textbf{Proposition~\ref{Proposition_LB_NA}}, we can obtain the following properties:

\enlargethispage{-0.03in}

\begin{property}\label{property_U}
    As the heterogeneity degree $U$ increases, the equivalent NA decreases. This is because, under the unified MC broadcast scheme mentioned in Section \ref{sub_sec_transmission}, greater heterogeneity inherently amplifies $\theta _{\varsigma}$ and $\sigma _{\varsigma}^2$. As indicated by \eqref{eq_rate_threshold_MSDC_LSDC}, any surges in $\theta _{\varsigma}$ and $\sigma _{\varsigma}^2$ demand stronger communication capability, i.e., a higher service rate, which degrades NA under given resources. By extension, any factors that demand stronger communication capability, such as more stringent QoS requirements, will inevitably lead to NA degradation.
\end{property}

\begin{property}\label{property_CPU2}
    Increasing computing resources ($\Omega_{\mathrm{R}}$) helps resolve MAC-layer packet accumulation, thereby reducing the processing delay violation probability. This contributes to decreasing the communication resources required to satisfy conditions \eqref{eq_rate_threshold_MSDC_LSDC}, ultimately enhancing NA. However, if communication resources, such as subchannels are scarce, adding only computing resources cannot resolve the PHY-layer transmission failures caused by dynamic channels. Consequently, the target NA cannot be achieved.
\end{property}

To address the issues in \textbf{Properties~\ref{property_U}--\ref{property_CPU2}}, we propose a cross-layer optimization that enables computing-communication resource interchange, which is detailed in the following.

\subsection{NA Enhancement}\label{Sec_NA_enhancement}
The fundamental mechanism lies in dynamically restructuring the delay composition. Under the delay constraints, adding computing resources allows for a reduced processing delay without increasing MAC-layer unreliability. Fundamentally, this reduction inversely extends the allowable transmission delay, equating to an expansion of time-frequency resources and thereby effectively converting computing resources into communication resources. However, this delay adjustment inherently alters the balance between processing delay violations and transmission failures. With augmented computing resources, optimally adjusting the processing and transmission delays under a given delay constraint maximizes reliability, thereby enhancing the probability of satisfying the QoS requirements for traffic under this constraint. Consequently, individually optimizing the delay compositions for mixed delay constraints effectively enhances NA. Therefore, we formulate the following optimization problem:
\begin{align}\label{eq_problem_delay}
{\mathop {\max }\limits_{ \mathbf{\Theta} } } \quad & \eta _{\mathrm{E}}^{\mathrm{H}}( {U,\mathbf{\Theta},N_{\replica} },{\Omega _{\mathrm{R}}} ) \\[-3pt]
{{\mathrm{s.t.}}}  \quad & \eqref{eq_delay_constraints_components}. \notag
\end{align}%
The optimal values of $\mathbf{\Theta}$, denoted by $\mathbf{\Theta} ^\star$ for the problem \eqref{eq_problem_delay}, can be obtained using the exhaustive search method. Since optimizing $\mathbf{\Theta}$ is element-wise decoupled, the independent search complexity per delay constraint—at a one-frame granularity—is given by ${\mathcal O} \big( {{\Nf^{\p}} \big( ( {D_{\max }^{\varsigma}} )_{u_{\varsigma}} )}  \big) $. Given realistic frame durations, computational delay is negligible.\footnote{Using MATLAB (Intel i5-1235U, 16GB RAM) with $T_{\f} = 0.1$ ms, the computational delay merely $\sim 0.0001$ ms for services under MSDC ($D^{\MSDC}_{\max} = 0.5$ ms) and $\sim 0.01$ ms for services under LSDC ($D^{\LSDC}_{\max} = 50$ ms), trivially satisfying $\ll D^{\varsigma}_{\max}$, $\varsigma \in \{\LSDC, \MSDC\}$.}


\addtolength{\topmargin}{0.1in}

\enlargethispage{-0.1in}

\section{Simulation and Numerical Results}\label{Sec_sim_num}
In this section, we validate our analysis and proposed optimization  through numerical simulations based on $10^{10}$ Monte-Carlo trials. We consider a service area with a radius $W_{\mathrm{D}}$ = $2,000$ m containing $40$ UEs, including $M_{\MSDC}$ = 5 UEs with unicast services under MSDC and $M_{\LSDC}$ = $35$ UEs with multicast services under LSDC. The number of FAPs is $L$ = 20. The aerial UE flight altitude range is ${\mathcal H} = [ {100,400} ]$ m. The $\kappa$-$\mu$ shadowed fading parameters are set as $\kappa \to 0$, $\mu= 3$, and $\bar{m} \to \infty$ to represent Nakagami-$m$ fading, which is widely used in aerial networks \cite{Wu2026TCOM}. The SNR loss due to wireless fronthaul and SNR gap due to FBL are set as $\phi_{\mathrm{wf}} = 0.8$ and $\phi_{\mathrm{fb}} = 1.5$, respectively. The service arrival process follows a Poisson distribution, where the average arrival rates for services under MSDC and LSDC are $20$ packets/s and $100$ packets/s, respectively \cite{Dong2021TWC}. The dOPL requirement is $\varepsilon _{\max }^{  \varsigma }$ = $10^{-5}$. The decoding error probability is $\varepsilon_{\decoding}^{\MSDC} = \frac{1}{4}\varepsilon _{\max }^{  \MSDC }$. The total delay bound under MSDC is $D_{\max }^{{\MSDC}}$ = 0.5 ms with relative velocity ${v_{\mathrm{r}}}$ = $30$ m/s \cite{Wu2026TCOM}. For services under LSDC, the total delay bounds $D_{\max }^{{\LSDC}}$ are $[50:5:100]$ ms with relative velocity ${v_{\mathrm{r}}}$ = $50$ m/s \cite{Wu2026TCOM}. The degree of heterogeneity $U$ = $Z$, $Z \ge 2$ indicates that the scenario contains services under MSDC and LSDC with $ D_{\max }^{{\LSDC}}( 1 ) \sim D_{\max }^{{\LSDC}}( {Z - 1} )$; $U$ = 1 refers to the scenario that contains services under MSDC. Different services may have different delay constraints. Under the unified MC broadcasting scheme, the arrival rates for downlink broadcasting are ${\theta_{\MSDC}}$ = $20 \times M_{\MSDC}$ packets/s and ${\theta_{\LSDC}}$ = $100 \times (U-1)$ packets/s with $\theta_{\varsigma}$ = $\sigma _\varsigma ^2$, $\varsigma \in \{\LSDC, \MSDC\}$. The target NA is ${\eta _{\max }} \to 1$, and whether it can be achieved is determined by evaluating whether the equivalent NA $\eta_{\mathrm{E}}^{\mathrm{H}} = 1$. The remaining parameters are listed in Table \ref{simulation_settings}.

\begin{table}[t]
\setlength{\abovecaptionskip}{0cm}
\caption{Simulation Parameters \cite{Wu2026TCOM, Dong2021TWC}}\label{simulation_settings}
\centering
\scriptsize
\begin{tabular}{|m{50mm}<\centering|m{15 mm}<\centering|}
\hline
\bfseries Parameter & \bfseries Value \\
\hline
\hline
\!\!\!Duration of each frame ${T_{\f}}$ (equals to TTI) \!\!\!& 0.1 ms \\
\hline
Duration of control signaling $T_{\mathrm{f}}^{( {\mathrm{c}} )}$ & 0.01 ms \\
\hline
Duration of data transmission $T_{\mathrm{f}}^{( {\mathrm{d}} )}$ & 0.09 ms \\
\hline
Subcarrier spacing $B_{0}$ & 15 kHz \\
\hline
Carrier frequency $f_{\mathrm{c}}$  & 2 GHz \\
\hline
Bandwidth of each subchannel $B$ & $ \lfloor \frac{{B_{\mathrm{C}}}}{{B_0}} \rfloor {B_0}$ \\
\hline
Bandwidth of each subchannel for LSDC/MSDC &  \!$B_{\LSDC}$/$B_{\MSDC}$\! = \!$\frac{1}{2}B$\! \\
\hline
Channel coherence bandwidth ${B_{\mathrm{C}}}$ & 0.5 MHz \\
\hline
Total bandwidth ${{B^{{\mathrm{tot}}}}}$ &  60 MHz \\
\hline
\!\!Number of independent subchannels $N_{\replica}$ \!\! & $\lfloor \frac{{{B^{{\mathrm{tot}}}}}}{{B_{\mathrm{C}}}} \rfloor$ \\
\hline
Noise power spectral density $N_{0}$ & \!\!\!-174 dBm/Hz\!\!\! \\
\hline
Transmit power ${p_{\mathrm{t}}} $/${p_{\mathrm{u}}} $ & 10/5 dBm \\
\hline
Packet size ${\varpi ^{\LSDC}}$/${\varpi ^{\MSDC}}$ & 1000/64 bits \\
\hline
FAP flight radius/altitude & 150/200 m \\
\hline
\end{tabular}
\vspace{-0.5em}
\end{table}

\begin{figure}[t]
\vspace{-0.3cm}
\centering
\begin{minipage}[t]{0.5\textwidth}
    \centering
    \includegraphics[width=0.9\textwidth]{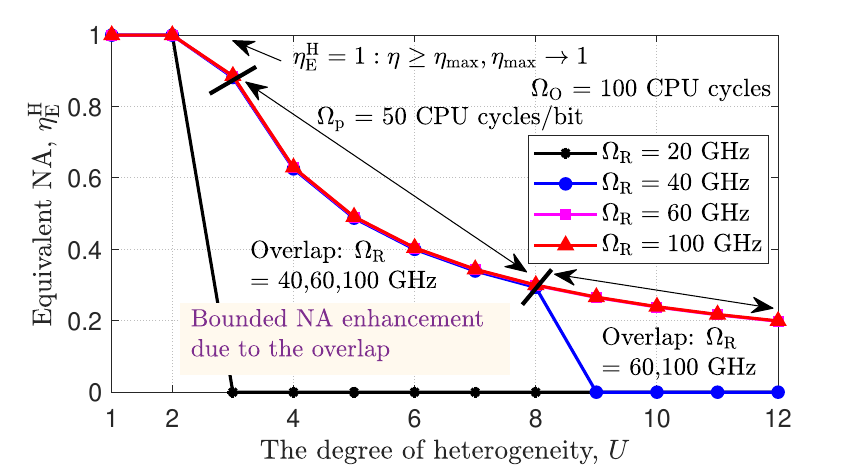}
    \vspace{-0.3em}
    \caption{Equivalent NA  vs. $U$ under ${D_{\p}^{\varsigma}}/{D_{\max}^{\varsigma}}$ = $\frac{2}{5}$ with different $\Omega_{\mathrm R}$.} \label{NA_vs_U}
\end{minipage}\\
\hfill
\begin{minipage}[t]{0.5\textwidth}
    \centering
    \includegraphics[width=0.9\textwidth]{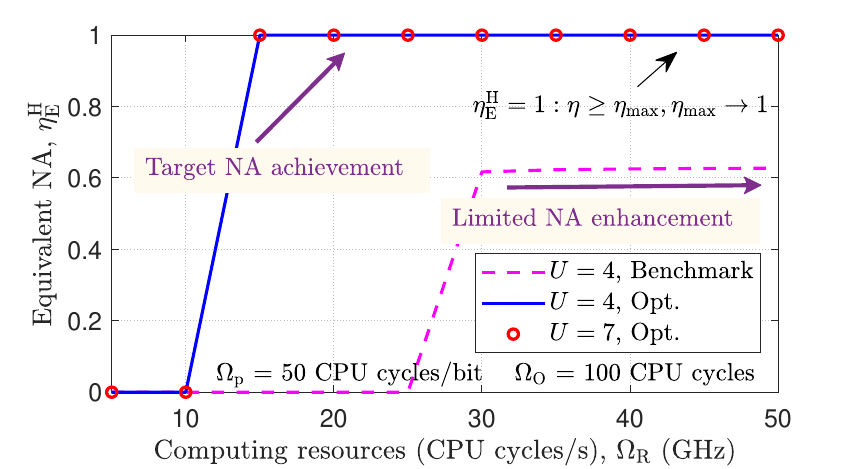}
    \vspace{-0.3em}
    \caption{Equivalent NA  vs. $\Omega_{\mathrm R}$ with different $U$.} \label{Comput_vs_NA_diff_U}
    \vspace{0.2em}
\end{minipage}
\end{figure}

Figure \ref{NA_vs_U} shows the equivalent NA $\eta^{\mathrm{H}}_{\E}$ versus the degree of heterogeneity $U$ under different computing resources $\Omega_{\mathrm R}$. As observed, $\eta^{\mathrm{H}}_{\E}$ decreases as $U$ increases. Furthermore, while increasing $\Omega_{\mathrm R}$ improves $\eta^{\mathrm{H}}_{\E}$ for a given $U$, this performance enhancement is bounded. This limitation is evident from the overlapping curves for different $\Omega_{\mathrm R}$ values. These results effectively verify \textbf{Properties~\ref{property_U}--\ref{property_CPU2}}.

\enlargethispage{-0.1in}

Figure \ref{Comput_vs_NA_diff_U} plots the equivalent NA $\eta^{\mathrm{H}}_{\E}$ against the computing resources (the number of CPU cycles/s, $\Omega_R$) under different heterogeneity degrees $U$. The legends ``Benchmark'' and ``Opt.'' represent the benchmark scheme with fixed $\frac{D_{\p}^{\varsigma}}{D_{\max}^{\varsigma}}$ = $\frac{2}{5}$ and the proposed cross-layer optimization, respectively. Compared with the benchmark, the proposed optimization enables the target NA to be achieved as the available computing resources increase for the same degree of service heterogeneity. More importantly, it can sustain the target NA even under a higher degree of heterogeneity, demonstrating that the proposed computing-communication resource interchange not only improves NA but also enlarges the range of heterogeneous services that can be reliably supported.

\enlargethispage{-0.15in}


\vspace{-0.1em}
\section{Conclusion}\label{Sec_Conclusion}
\vspace{-0.4em}
In this paper, we addressed the network unavailability in LA-HetNets due to heterogeneity and cross-regional traffic. Our analysis revealed that supplementing computing resources improves NA by resolving MAC-layer unreliability. However, under greater heterogeneity, merely increasing computing resources is insufficient to achieve the target NA since NA degrades sharply with extended heterogeneity due to amplified communication capability demands. To overcome this, we proposed a cross-layer optimization method to realize computing-communication resource interchange. Simulation results validated that adopting this resource interchange approach can effectively enhance NA to achieve its target value and support significantly greater heterogeneity than conventional computing-only expansion strategies.

\vspace{-0.1em}
\appendices
\section{Proof of Proposition \ref{Proposition_LB_NA}}
\label{appendix_Proposition_LB_NA}
\vspace{-0.5em}
For given ${D_{\dd}^{\varsigma}}$, ${D_{\p}^{\varsigma}}$, and ${D_{\dd}^{\varsigma}} + {D_{\p}^{\varsigma}} = D_{\max }^{\varsigma}$, under the unified MC broadcasting scheme with $N_{\replica}$ subchannels for transmitting packet replicas, the transmission failure probability of the $m$th UE based on the effective achievable rate can be given by \cite{Li2024} 
\begin{equation}\label{eq_outage_prob}
\setlength{\abovedisplayskip}{0pt}
\setlength{\belowdisplayskip}{0pt}
    {\varepsilon _{\dd}^{\varsigma}}( {{D_{\dd}^{\varsigma}}} ) \triangleq \Big[{{\mathbb{P}_{N_{\varsigma}}}\big\{ {N_{\f}^{\dd}( {{D_{\dd}^{\varsigma}}} )R^{\varsigma} < N_{\varsigma}} \big\}}+ \varepsilon_{\decoding}^{\varsigma}\Big]^{N_{\replica}},
\end{equation}
where ${{\mathbb{P}_{N_{\varsigma}}}\big\{ {N_{\f}^{\dd}( {{D_{\dd}^{\varsigma}}} )R^{\varsigma} < N_{\varsigma}} \big\}}$ denotes the transmission outage probability \cite{Alsenwi2019CL}, while  $N_{\varsigma}$ and $ N_{\f}^{\dd}( {{D_{\dd}^{\varsigma}}} )R^{\varsigma}$ represent the number of random arrival packets within ${D_{\max}^{\varsigma}}$ and the number of successfully transmitted packets within ${D_{\dd}^{\varsigma}}$, respectively. For given ${D_{\p}^{\varsigma}}$, the processing delay violation probability can be expressed as \cite{She2019IoTJ}
\begin{equation}\label{eq_processing_delay_prob_LSDC}
\setlength{\abovedisplayskip}{0pt}
\setlength{\belowdisplayskip}{0pt}
    {\varepsilon _{\p}^{\varsigma}}( {{D_{\p}^{\varsigma}}} ) \triangleq {\mathbb{P}_{\Phi _{\varsigma} }}\Big( {{D_{\p}^{\varsigma}} < \frac{N_{\replica}{\Phi _{\varsigma} ( {{\Omega _{\p}}{\varpi ^{{\varsigma}}} + {\Omega _{\mathrm{O}}}} )}}{{{\Omega _{\mathrm{R}}}}}{T_{\f}}} \Big),
\end{equation}
where $\Phi _{\varsigma} \! =\! N^{\varsigma,\p}_{\LSDC} + N^{\varsigma,\p}_{\MSDC}$ is the total number of arriving packets, while $N^{\varsigma,\p}_{\LSDC}$ and $N^{\varsigma,\p}_{\MSDC}$ are the numbers of arriving packets within ${D_{\p }^{\varsigma}}$ for services under LSDC and MSDC, respectively.

According to Chebyshev's inequality, for a RV $X$ with mean $\bar X$ and variance $\sigma _X^2$, the probability that $X$ is greater than a threshold $A$ can be bounded by \cite{LiuTCOM2025}
\begin{equation}\label{eq_appendix_Chebyshev_inequality_1}
\setlength{\abovedisplayskip}{0pt}
\setlength{\belowdisplayskip}{0pt}
    {\mathbb{P}_X}\{ {A \le X} \} \le {\sigma _X^2}/{{{( {A - \bar X} )}^2}}, A > \bar X.
\end{equation}

Then, we can obtain the UB on ${\varepsilon _{\p}^{\varsigma}}( {{D_{\p}^{\varsigma}}} )$ as ${\hat \varepsilon_{\p}^{\varsigma}}( {{D_{\p}^{\varsigma}}} )$, which is expressed in \eqref{eq_processing_delay_prob_UB}. Similarly, ${\varepsilon _{\dd}^{\varsigma}}( {{D_{\dd}^{\varsigma}}} )$ can be bounded by
\begin{equation}\label{eq_transmission_failure_prob_UB}
\setlength{\abovedisplayskip}{0pt}
\setlength{\belowdisplayskip}{0pt}
    \!{\hat \varepsilon_{\dd} ^{\varsigma}}( {{D_{\dd} ^{\varsigma}}} ) \!=\! \Bigg[\! {\frac{{N_{\f}^{\varsigma}(D_{\max}^{\varsigma})\sigma _{\varsigma}^2}}{{{{\big( {N_{\f}^{\dd}( {{D_{\dd}^{\varsigma}}} )R^{\varsigma} - N_{\f}^{\varsigma}(D^{\varsigma}_{\max}){\theta _{\varsigma}}} \big)}^2}}}} + \varepsilon_{\decoding}^{\varsigma} \Bigg]^{\!N_{\replica}}\!\!\!.\!
\end{equation}

If ${\hat \varepsilon_{\p}^{\varsigma}}( {{D_{\p}^{\varsigma}}} ) < \varepsilon_{\max}^{\varsigma}$, the QoS requirements ($D_{\max}^{\varsigma}, \varepsilon_{\max}^{\varsigma}$) can be satisfied under ${\varepsilon_{\dd} ^{\varsigma}}( {{D_{\dd} ^{\varsigma}}} ) \le \varepsilon_{\max}^{\varsigma} - {\hat \varepsilon_{\p}^{\varsigma}}( {{D_{\p}^{\varsigma}}} )$. Then, the QoS requirements ($D_{\max}^{\varsigma}, \varepsilon_{\max}^{\varsigma}$) can be satisfied under ${\hat \varepsilon_{\dd} ^{\varsigma}}( {{D_{\dd} ^{\varsigma}}} ) \le \varepsilon_{\max}^{\varsigma} - {\hat \varepsilon_{\p}^{\varsigma}}( {{D_{\p}^{\varsigma}}} )$. From \eqref{eq_transmission_failure_prob_UB}, the condition for satisfying the QoS requirements ($D_{\max}^{\varsigma}, \varepsilon_{\max}^{\varsigma}$) can be expressed as
\begin{align}\label{eq_condition_LSDC_MSDC}
& R^{\varsigma}  \ge
 \frac{{N_{\f}^{\varsigma}}(D_{\max }^{\varsigma})}{N_{\f}^{\dd}(D_{\dd}^{\varsigma})} \\
 & \times  \Bigg[{{\theta _{\varsigma}} + \Bigg(\!\frac{{\sigma _{\varsigma}^2}}{{N_{\f}^{\varsigma}}(D_{\max }^{\varsigma})\big({{{\big( {{\varepsilon^{\varsigma} _{\max }} - {\hat \varepsilon_{\p}^{\varsigma}}( {{D_{\p}^{\varsigma}}} ) } \big)}^{\frac{1}{{{N_{\replica}}}}} - \varepsilon_{\decoding} ^{\varsigma} }}\big)}}\!\Bigg)^{\frac{1}{2}}\Bigg], \notag
\end{align}
where the right-hand side of \eqref{eq_condition_LSDC_MSDC} is defined as the threshold of service rate to satisfy the QoS requirements ($D_{\max}^{\varsigma}, \varepsilon_{\max}^{\varsigma}$), denoted by $R_{\mathrm{th}}^{\varsigma}(D_{\p}^\varsigma,D_{\dd}^\varsigma)$. Note that $R^{\varsigma}$ remains constant regardless of the mixed delay constraints. However, $R_{\mathrm{th}}^{\varsigma}\big( (D_{\p}^\varsigma)_{u_{\varsigma}},(D_{\dd}^\varsigma)_{u_{\varsigma}} \big)$ generally varies across different delay constraints. To ensure the QoS requirements for all UEs (i.e., $\eta \ge {\eta _{\max }},{\eta _{\max }}  \to 1$), $R^{\varsigma}$ must satisfy the maximum rate threshold for both LSDC and MSDC, where the maximum rate threshold is given by
\begin{align}
\!  R_{\mathrm{th},\max}^{\varsigma} = \max \big\{ {R_{\mathrm{th}}^{\varsigma}\big( (D_{\p}^\varsigma)_{u_{\varsigma}},(D_{\dd}^\varsigma)_{u_{\varsigma}} \big)}, \forall u_{\varsigma} \in [1, U_{\varsigma}] \big\}.\!
\end{align}


Therefore, ($\eta \ge {\eta _{\max }},{\eta _{\max }}  \to 1$) can be guaranteed when the following condition is satisfied:
\begin{align}\label{eq_LB_all_services_LSDC_MSDC}
\!\!  \prod\nolimits_{\varsigma  \in \{ {\LSDC,\MSDC} \}} \!\! {\min \bigg(1, \frac{R^{\varsigma}}{ R_{\mathrm{th},\max}^{\varsigma} }\bigg)} = 1.
\end{align}

Motivated by this rigorous physical equivalence, the product metric on the left-hand side of \eqref{eq_LB_all_services_LSDC_MSDC} can be regarded as the equivalent NA to evaluate whether the target NA is achieved. Consequently, we obtain \eqref{eq_NA_all_REC}, and prove \textbf{Proposition \ref{Proposition_LB_NA}}.

\enlargethispage{-0.15in}

\vspace{-0.3em}
\bibliographystyle{IEEEtran}
\bibliography{myref}

\end{document}